\documentclass[11pt]{article}

\usepackage{fullpage}
\usepackage{amsmath}
\usepackage{xspace}
\usepackage{amssymb}
\usepackage{epsfig}
\usepackage{color}

\newtheorem{Thm}{Theorem}
\newtheorem{Lem}[Thm]{Lemma}
\newtheorem{Cor}[Thm]{Corollary}

\newtheorem{Def}{Definition}

\newenvironment{proof}{\noindent {\textbf{Proof }}}{$\Box$ \medskip}

\mathchardef\mhyphen="2D
\newcommand{\defeq}{\stackrel{\mathsf{def}}{=}}
\newcommand{\ket}[1]{| #1 \rangle}
\newcommand{\bra}[1]{\langle #1 |}

\newcommand\mbR{\mbox{$\mathbb{R}$}}
\newcommand\mbC{\mbox{$\mathbb{C}$}}

\newcommand{\suppress}[1]{}

\newcommand\cptp{\mbox{\sf {CPTP}}\xspace}

\begin{document}
\title{\bf On characterizing quantum correlated equilibria}
\author{Zhaohui Wei \thanks{Center for Quantum Technologies, National University of Singapore. Email: {\tt cqtwz@nus.edu.sg}} \qquad Shengyu Zhang\thanks{Department of Computer Science and Engineering and The Institute of Theoretical Computer Science and Communications, The Chinese University of Hong Kong, Shatin, N.T., Hong Kong. Email: {\tt syzhang@cse.cuhk.edu.hk}.}}
\maketitle

\begin{abstract}
Quantum game theory lays a foundation for understanding the interaction of people using quantum computers with conflicting interests. Recently Zhang proposed a simple yet rich model to study quantum strategic games, and addressed some quantitative questions for general games of growing sizes~\cite{Zha10}. However, one fundamental question that the paper did not consider is the characterization of quantum correlated equilibria (QCE). In this paper, we answer this question by giving a sufficient and necessary condition for an arbitrary state $\rho$ being a QCE. In addition, when the condition fails to hold for some player $i$, we give an explicit POVM for that player to achieve a strictly positive gain. Finally, we give some upper bounds for the maximum gain by playing quantum strategies over classical ones, and the bounds are tight for some games.
\end{abstract}

\section{Introduction}
Game theory studies the interaction of different players with possibly conflicting goals~\cite{OR94,FT91,VNRT07}. Equilibrium is a central solution concept which characterizes the situation that no player likes to deviate from the current strategy provided that all other players do not change theirs. If each player chooses her strategy from a distribution (on her own strategy space), then the joint product distribution is a (mixed) Nash equilibrium if no player has incentive to change her distribution. A fundamental theorem by Nash says that any any game with a finite set of strategies has at least one Nash equilibrium~\cite{Nas51}.

Aumann~\cite{Aum74} gave an important generalization of Nash equilibrium, called correlated equilibrium (CE), where a Referee selects a joint strategy $s=(s_1,\ldots,s_k)$ from some distribution $p$ and suggested $s_i$ to the $i$-th player. The joint distribution $p$ is a correlated equilibrium (CE) if no player $i$, when sees only  her part $s_i$, cannot deviate from this suggested strategy to improve her expected payoff. The notion of correlated equilibria captures the optimal solution in natural games such as the Traffic Light and the Battle of the Sexes (\cite{VNRT07}, Chapter 1). The set of CE also has good mathematical properties, such as being convex, with Nash equilibria being some of the vertices of the polytope. Computationally, it is benign for finding the best CE (of any game with constant players), measured by any linear function of payoffs, simply by solving a linear program. Other example include that a natural learning dynamics can lead to an approximate variant of CE (\hspace{-.08em}\cite{VNRT07}, Chapter 4), and all CE in a graphical game with $n$ players and $\log(n)$ maximum degree can be found in polynomial time (\hspace{-.08em}\cite{VNRT07}, Chapter 7).

In the quantum world, quantum game theory lays a foundation for understanding the interaction of people using quantum computers with conflicting interests. Indeed, quantization of classical strategic games have drawn much attention in the past decade. Despite the rapid accumulation of literature~\cite{EWL99,BH01b,LJ03,FA03,FA05,DLX+02b,DLX+02a,PSWZ07}, the whole picture of the area is not as clean as one
desires, partially due to controversy in
models~\cite{BH01,CT06}.
Recently, Zhang proposed a new model which is simpler, arguably more
natural, and corresponding more precisely to the classical strategic
games~\cite{Zha10}; also see that paper for a review of the existing
literature under the name of ``quantum games". Other than the model,
what mainly distinguishes that work from previous ones is the
generality of the classical game it studies: Unlike previous work
focusing on specific games of small sizes or refereed extensive
games, the paper studies general strategic games of growing sizes.
In addition, rather than aiming at qualitative questions such as
whether playing quantum strategies has any advantage as in previous
work,~\cite{Zha10} studies quantitative questions such as how much
quantum advantage can a general-sized game have.

Solution concepts such as Nash equilibrium and correlated
equilibrium are naturally extended to the quantum model
in~\cite{Zha10}, and it is studied how well standard maps between
classical and quantum states preserves equilibrium properties. It
turns out that if $p$ is a classical Nash equilibrium, then both
$\rho = \sum_s p(s) \ket{s}\bra{s}$ and $\ket{\psi} = \sum_s
\sqrt{p(s)} \ket{s}$ are quantum Nash equilibria. But correlated
equilibrium is of a different story: While $\rho = \sum_s p(s)
\ket{s}\bra{s}$ is still guaranteed to be a quantum correlated
equilibrium if $p$ is a classical correlated equilibrium, the
mapping of $p \mapsto \ket{\psi} = \sum_s \sqrt{p(s)} \ket{s}$ can
severely destroy the correlated equilibrium property. Therefore,
correlated equilibrium is a subtler subject to study.

Given the importance of correlated equilibrium in game theory and computer science, it is desirable to well understand quantum correlated equilibria. One fundamental question that~\cite{Zha10} did not address is the following: For an arbitrary (classical) strategic game, can we characterize all the quantum correlated equilibria (QCE) in the quantum game?

In this paper, we answer this question by giving the following sufficient and necessary condition for any given game and any state $\rho$.
\begin{Thm}\label{thm: main}
A quantum state $\rho$ in space $H$ is a QCE if and only if for each
player $t$, when we write
$\rho=[\rho_{j_1j_2}^{i_1i_2}]_{i_1j_1,i_2j_2}$, where
$i_1,i_2\in[m]\defeq\{1,2,...,m\}$ and $j_1,j_2\in[n]$ with $m =
dim(H_t)$ and $n = dim(H_{-t})$, we have
\begin{equation}\label{eq: sdp condition}
B_i \defeq \Big[\sum_j\rho^{i_1i_2}_{jj}(a_{ij}-a_{i_1j})\Big]_{i_1i_2} \preceq 0, \quad \forall i\in [m].
\end{equation}
\end{Thm}
We first show that the condition is sufficient by working out the gain of the POVM $\{E_i\}$ as $\sum_i tr(E_iB_i)$, and then give two different proofs for the necessity part. The first one is based on semi-definite program, which is simple yet not intuitive enough. The second proof is constructive, which shows that when the condition fails to hold for some player $i$, one can find an explicit POVM for that player to achieve a strictly positive gain. Finally, for those cases that are not QCE, we give some upper bounds for the maximum gain, which are shown to be tight for some games.

The paper is organized as follows. Some preliminary knowledge is introduced in Section \ref{sec: pre}. In Section \ref{sec: condition} we give the sufficient and necessary condition, and in Section \ref{sec: explicit} the necessity part is reproved constructively, which can be regarded as the operational explanation of this condition. In Section \ref{sec: ub}, we obtain some upper bounds of the gain when $\rho$ is not a QCE. Some open problems are listed in Section \ref{sec: open problems}.

\section{Definitions and notation}\label{sec: pre}
A matrix $A\in \mbC^{n\times n}$ is a Hermitian if $A^\dag = A$, or equivalently, $A$ has a spectral decomposition and all eigenvalues are real numbers. A matrix $A\in \mbC^{n\times n}$ is a positive (semi-definite), written as $A\succeq 0$, if $\bra{\psi}A\ket{\psi} \geq 0$ for all column vectors $\ket{\psi}$. Equivalently, $A\succeq 0$ if and only if $A$ has a spectral decomposition and all eigenvalues are nonnegative numbers. Thus all positive matrices are Hermitians. Define $A\preceq 0$ if $-A\succeq 0$.

\subsection{Classical strategic games}
Suppose that in a classical game there are $k$ players. Each player $i$ has a set $S_i$ of strategies. To play the game, each player $i$ selects a strategy $s_i$ from $S_i$. We use $s=(s_1,\ldots, s_k)$ to denote the \emph{joint strategy} selected by the players and $S= S_1 \times \ldots \times S_k$ to denote the set of all possible joint strategies. Each player $i$ has a utility function $u_i: S \rightarrow \mbR$, specifying the \emph{payoff} or \emph{utility} $u_i(s)$ to player $i$ on the joint strategy $s$. We use subscript $-i$ to denote the set $[k]-\{i\}$, so $s_{-i}$ is $(s_1, \ldots, s_{i-1}, s_{i+1}, \ldots, s_k)$.

\begin{Def}
A \emph{pure Nash equilibrium} is a joint strategy $s = (s_1, \ldots ,s_k) \in S$ satisfying that
\begin{align*}
    u_i(s_i,s_{-i}) \geq  u_i(s_i',s_{-i}), \qquad \forall i\in [k], \forall s'_i\in S_i.
\end{align*}
A \emph{(mixed) Nash equilibrium (NE)} is a product probability distribution $p = p_1 \times \ldots \times p_k$, where each $p_i$ is a probability distributions over $S_i$, satisfying that
\begin{align*}
    \sum_{s_{-i}} p_{-i}(s_{-i}) u_i(s_i,s_{-i}) \geq  \sum_{s_{-i}} p_{-i}(s_{-i}) u_i(s_i',s_{-i}), \quad \forall i\in [k], \ \forall s_i, s'_i\in S_i \text{ with } p_i(s_i)>0.
\end{align*}
\end{Def}

There are various extensions of (mixed) Nash equilibria. Aumann~\cite{Aum74} introduced a relaxation called \emph{correlated equilibrium}. This notion assumes an external party, called Referee, to draw a joint strategy $s = (s_1, ..., s_k)$ from some probability distribution $p$ over $S$, possibly correlated in an arbitrary way, and to suggest $s_i$ to Player $i$. Note that Player $i$ only sees $s_i$, thus the rest strategy $s_{-i}$ is a random variable over $S_{-i}$ distributed according to the conditional distribution $p|_{s_i}$, the distribution $p$ conditioned on the $i$-th part being $s_i$. Now $p$ is a correlated equilibrium if any Player $i$, upon receiving a suggested strategy $s_i$, has no incentive to change her strategy to a different $s_i' \in S_i$, assuming that all other players stick to their received suggestion $s_{-i}$.

\begin{Def} \label{thm:CE}
A \emph{correlated equilibrium (CE)} is a probability distribution $p$ over $S$ satisfying that
\begin{align*}
    \sum_{s_{-i}} p(s_i,s_{-i}) u_i(s_i,s_{-i}) \geq  \sum_{s_{-i}} p(s_i,s_{-i}) u_i(s_i',s_{-i}), \qquad \forall i\in [k], \ \forall s_i, s'_i\in S_i.
\end{align*}
\end{Def}
Notice that a correlated equilibrium $p$ is an Nash equilibrium if $p$ is a product distribution.

\subsection{Quantum strategic games}

In this paper we consider quantum games which allows the players to use strategies quantum mechanically. We assume the basic background of quantum computing; see~\cite{NC00} and~\cite{Wat08} for comprehensive introductions. The set of admissible super operators, or equivalently the set of completely positive and trace preserving (CPTP) maps, of density matrices in Hilbert spaces $H_A$ to $H_B$, is denoted by $\cptp(H_A, H_B)$. We write $\cptp(H)$ for $\cptp(H,H)$.

For a strategic game being played quantumly, each player $i$ has a Hilbert space $H_i = span\{s_i: s_i\in S_i\}$, and a joint strategy can be any quantum state $\rho$ in $H = \otimes_i H_i$. The players are supposed to measure the state $\rho$ in the computational basis, giving a distribution over the set $S$ of classical joint strategies, and yielding a payoff for each player. Therefore the (expected) payoff for player $i$ on joint strategy $\rho$ is
\begin{equation}
    u_i(\rho) = \sum_s \bra{s} \rho \ket{s} u_i(s).
\end{equation}
Please refer to~\cite{Zha10} for more explanations of the model.

Corresponding to changing strategies in a classical game, now each player $i$ can apply an arbitrary CPTP operation on $H_i$. So the natural requirement for a state being a quantum Nash equilibrium is that each player cannot gain by applying any admissible operation on her strategy space. The concepts of quantum Nash equilibrium, and quantum correlated equilibrium, and quantum approximate equilibrium are defined in the following, where we overload the notation by writing $\Phi_i$ for $\Phi_i \otimes I_{-i}$ if no confusion is caused.

\begin{Def}
A \emph{quantum Nash equilibrium (QNE)} is a quantum strategy $\rho = \rho_1 \otimes \cdots \otimes \rho_k$ for some mixed states $\rho_i$'s on $H_i$'s satisfying that
\begin{align*}
    u_i(\rho) \ge u_i(\Phi_i(\rho)), \qquad \forall i\in [k], \ \forall \Phi_i \in \cptp(H_i).
\end{align*}
\end{Def}

\begin{Def} \label{def:QCE}
A \emph{quantum correlated equilibrium (QCE)} is a quantum strategy $\rho$ in $H$ satisfying that
\begin{align*}
    u_i(\rho) \ge u_i(\Phi_i(\rho)), \qquad \forall i\in [k], \ \forall \Phi_i \in \cptp(H_i).
\end{align*}
An \emph{$\epsilon$-approximate quantum correlated equilibrium ($\epsilon$-QCE)} is a quantum strategy $\rho$ in $H$ satisfying that
\begin{align*}
    u_i(\Phi_i(\rho)) \leq u_i(\rho) + \epsilon, \qquad \forall i\in [k], \forall \Phi_i \in \cptp(H_i).
\end{align*}
\end{Def}

When we later characterize quantum correlated equilibrium, we will need that no player can increase her payoff, so a condition is required for each player. For easy presentation, we fix an arbitrary player, say, Player 1, and consider the possible increase of her payoff by local operations. Write the state as
\[
\rho = \big[\rho_{j_1j_2}^{i_1i_2}\big]_{i_1j_1,i_2j_2},
\]
where $i_1,i_2\in H_1$ and $j_1,j_2\in H_{-1}$. Suppose that the dimensions of $H_1$ and $H_{-1}$ are $m$ and $n$, respectively.

\section{Characterization of quantum correlated equilibrium}\label{sec: condition}
We will first give an explicit expression of the gain of Player 1 applying the POVM $\{E_i\}$ (compared to the measurement in the computational basis).
\begin{Lem}\label{lem: gain}
Suppose Player $1$ uses the POVM measurement $E = \{E_i\}$ and other players use $M = \{\ket{j}\bra{j}\}$ to measure their parts in the computational basis, then the gain of Player 1 by applying $E$ than measuring in the computational basis is
\[
Gain \defeq u_1\big((E\otimes M)\rho\big) - u_1(\rho) = \sum_i tr(E_i B_i), \quad \text{ where } B_i = \Big[\sum_j\rho^{i_1i_2}_{jj}(a_{ij}-a_{i_1j})\Big]_{i_1i_2}.
\]
\end{Lem}
\begin{proof}
The probability of getting $(i,j)$ is $tr((E_i\otimes \ket{j}\bra{j}) \rho) = \sum_{i_1,i_2} E_i(i_1,i_2)^*\rho_{jj}^{i_1i_2}$. Note that
\begin{align}
\sum_{ij} \sum_{i_1i_2} E_i(i_1,i_2)^* \rho_{jj}^{i_1i_2} a_{i_1j} & = \sum_{i_1i_2j} \big(\sum_{i} E_i(i_1,i_2)^*\big) \rho_{jj}^{i_1i_2} a_{i_1j} = \sum_{i_1j} \rho_{jj}^{i_1i_1} a_{i_1j} = \sum_{ij} \rho_{jj}^{ii} a_{ij}
\end{align}
where the second equality is because being a POVM measurement, $\{E_i\}$ satisfies $\sum_{i} E_i(i_1,i_2) = \delta_{i_1,i_2}$. Therefore,
\begin{align}
Gain & = \sum_{ij}\Big(\sum_{i_1,i_2} E_i(i_1,i_2)^*\rho_{jj}^{i_1i_2} - \rho_{jj}^{ii}\Big) a_{ij} \\
& = \sum_{ij}\sum_{i_1,i_2} E_i(i_1,i_2)^*\rho_{jj}^{i_1i_2}a_{ij} - \sum_{ij}\sum_{i_1,i_2} E_i(i_1,i_2)^*\rho_{jj}^{i_1i_2} a_{i_1j} \\
& = \sum_i \sum_{i_1,i_2} E_i(i_1,i_2)^* \sum_j \rho_{jj}^{i_1i_2} (a_{ij} - a_{i_1j}) \\
& = \sum_i tr(E_i^\dag B_i) = \sum_i tr(E_i B_i)
\end{align}
\end{proof}

The above lemma immediate gives a sufficient condition for a state $\rho$ being a QCE.
\begin{Thm}\label{thm: sufficient}
    If for each player, the corresponding $B_i \preceq 0$ for all $i\in [m]$, then $\rho$ is a QCE.
\end{Thm}
\begin{proof}
By the above lemma, the gain $\sum_i tr(E_i B_i) \leq 0$ because each $B_i \preceq 0$ and each $E_i \succeq 0$. Since this holds for all possible POVM $\{E_i\}$, $\rho$ is a QCE by definition.
\end{proof}

Next we will use SDP duality to show that the condition is also necessary.
\begin{Thm}\label{thm: necessary}
Suppose $\rho$ is a QCE, then for each player $t$, when we write $\rho=[\rho_{j_1j_2}^{i_1i_2}]_{i_1j_1,i_2j_2}$, where $i_1,i_2\in[m]$ and $j_1,j_2\in[n]$ with $m = dim(H_t)$ and $n = dim(H_{-t})$, we have
\begin{equation}\label{eq: sdp condition}
B_i \defeq \Big[\sum_j\rho^{i_1i_2}_{jj}(a_{ij}-a_{i_1j})\Big]_{i_1i_2} \preceq 0, \quad \forall i\in [m].
\end{equation}
\end{Thm}
\begin{proof}
Since $\rho$ is a QCE, Player $t$ cannot increase her payoff by applying any POVM. Therefore, the value of the following maximization problem
\begin{align*}
    \max & \quad \sum_{i} tr(E_iB_i)\\
    \emph{s.t.} & \quad E_i \succeq 0, \quad \forall i\in [m], \\
    & \quad \sum_{i=1}^m E_i = I_m
\end{align*}
is equal to 0. The dual of the SDP is the following.
\begin{align*}
    \textbf{\emph{Dual:}} \\
    \min & \quad \emph{tr}(Y)\\
    \emph{s.t.} & \quad Y \succeq B_i,\ \forall i\in [m]
\end{align*}
Note that $Y(i,i) \geq B_i(i,i) = 0$, so the optimal value being 0 implies that all $Y(i,i) = B_i(i,i) = 0$. But $Y - B_i \succeq 0$, so actually the $i$-th row of $Y - B_i$ is all $0$. Since the $i$-th row of $B_i$ is 0, so the $i$-th row of $Y_i$ is 0. Thus the entire $Y = 0$, giving the claimed relation $B_i \preceq 0$.
\end{proof}

Since a negative matrix is a Hermitian, an immediate corollary is as follows. Why this corollary
is valid has an operational explanation, which will be shown in the
next section.
\begin{Cor}
    If $\rho$ is a QCE, then
    \begin{equation}
    \sum_j\rho^{i_1i_2}_{jj}a_{i_1j} = \sum_j\rho^{i_1i_2}_{jj} a_{i_2j}, \quad \forall i_1,i_2\in [m].
    \label{eq: eq condition}
    \end{equation}
\end{Cor}
Both necessary conditions in Theorem \ref{thm: necessary} and the above corollary are not constructive in the sense that if $\rho$ is not a QCE, they do not provide an explicit POVM to realize a strictly positive gain of
payoff. We will resolve this issue in the next section.

\section{A constructive proof of the characterization}\label{sec: explicit}
In last section, we give two necessary conditions Eq.~\eqref{eq: sdp condition} and Eq.~\eqref{eq: eq condition}, the first of which is also sufficient (while the second is not by itself). In this section, we will give explicit local operations to increase the payoff if these conditions are not satisfied. We will first study in Section \ref{sec: eq condition} the scenario that Eq.~\eqref{eq: eq condition} is violated, in which case a local \emph{unitary} operation is explicitly given to achieve a positive gain. Based on this result, we will then consider in Section \ref{sec: sdp condition} the general scenario of Eq.~\eqref{eq: sdp condition} being violated, in which case we will exhibit an explicit POVM with a positive gain for the player.

\subsection{Eq.~\eqref{eq: eq condition} violated: gain by an explicit local unitary}\label{sec: eq condition}
\begin{Lem}\label{lem:nece-eq}
If $\sum_{j}\rho^{i_1i_2}_{jj}a_{i_1j} \neq \sum_{j}\rho^{i_1i_2}_{jj}a_{i_2j}$
for some $i_1,i_2\in [m]$, then there exists an explicit unitary only on $span\{\ket{i_1},\ket{i_2}\}$ to make an increase of payoff for the player.
\end{Lem}
\begin{proof}
Consider the unitary operator $U$ defined by
\begin{align*}
U|i_1\rangle=u_{11}|i_1\rangle+u_{12}|i_2\rangle, \\
U|i_2\rangle=u_{21}|i_1\rangle+u_{22}|i_2\rangle.
\end{align*}
The new probability distribution of strategy after the operation of
$U^\dag$ on $span\{\ket{i_1},\ket{i_2}\}$ and identity on other
$i$'s is
\begin{equation}
p_{ij}=Tr((U|i\rangle\langle
i|U^\dag \otimes|j\rangle\langle j|)\rho).
\end{equation}
where we overload the notation by writing $U$ for $U\otimes
I_{[m]-\{i_1,i_2\}}$. Note that when $i\in[m]-\{i_1,i_2\}$,
$p_{ij}=Tr(|i\rangle\langle i|\otimes|j\rangle\langle
j|)\rho)=\rho^{ii}_{jj}$. Thus, the gain of Player 1 by the
operation $\Psi$ is
\begin{eqnarray*}
Gain & = & \sum_{ij}\Big(p_{ij}-\rho^{ii}_{jj}\Big)a_{ij} \\
& = & \sum_{j}\Big(p_{i_1j}-\rho^{i_1i_1}_{jj}\Big)a_{i_1j}+\sum_{j}\Big(p_{i_2j}-\rho^{i_2i_2}_{jj}\Big)a_{i_2j}\\
& = & \sum_{j}\Big(\sum_{a,b=1}^2\rho_{jj}^{i_a i_b} u_{1,i_a}^*u_{1,i_b} - \rho^{i_1i_1}_{jj}\Big)a_{i_1j} + \sum_{j}\Big(\sum_{a,b=1}^2\rho_{jj}^{i_ai_b} u_{2,i_a}^*u_{2,i_b} - \rho^{i_2i_2}_{jj}\Big)a_{i_2j}\\
& = &
\sum_{j}\Big(|u_{11}|^2\rho^{i_1i_1}_{jj}+u_{11}^{*}u_{12}\rho^{i_1i_2}_{jj}+u_{11}u_{12}^{*}\rho^{i_2i_1}_{jj}+|u_{12}|^2\rho^{i_2i_2}_{jj}-\rho_{jj}^{i_1i_1}\Big)a_{i_1j}\\
& & +
\sum_{j}\Big((|u_{21}|^2\rho^{i_1i_1}_{jj}+u_{21}^{*}u_{22}\rho^{i_1i_2}_{jj}+u_{21}u_{22}^{*}\rho^{i_2i_1}_{jj}+|u_{22}|^2\rho^{i_2i_2}_{jj}-\rho_{jj}^{i_2i_2}\Big)a_{i_2j}.
\end{eqnarray*}
Since $U$ is a unitary operation, we have
\begin{equation}
u_{11}^{*}u_{12}+u_{21}^{*}u_{22}=0, \ \
u_{11}u_{12}^{*}+u_{21}u_{22}^{*}=0,
\end{equation}
and
\begin{equation}
|u_{11}|^2+|u_{21}|^2=1, \ \ |u_{12}|^2+|u_{22}|^2=1.
\end{equation}
Thus, we obtain
\begin{eqnarray*}
Gain & = & u_{11}^*u_{12}\sum_{j}\Big(a_{i_1j}-a_{i_2j}\Big)\rho^{i_1i_2}_{jj}+u_{11}u_{12}^*\sum_{j}\Big(a_{i_1j}-a_{i_2j}\Big)\rho^{i_2i_1}_{jj}\\
& & - |u_{12}|^2\sum_{j}\Big(a_{i_2j}-a_{i_1j}\Big)\rho_{jj}^{i_2i_2}-|u_{21}|^2\sum_{j}\Big(a_{i_1j}-a_{i_2j}\Big)\rho_{jj}^{i_1i_1}\\
\end{eqnarray*}
Since
$\sum_{j}(a_{i_1j}-a_{i_2j})\rho^{i_1i_2}_{jj}\neq 0$, we have  $\sum_{j}(a_{i_1j}-a_{i_2j})\rho^{i_2i_1}_{jj} = \sum_{j}(a_{i_1j}-a_{i_2j})(\rho^{i_1i_2}_{jj})^* \neq 0$ as well. Define a positive real number $c$ by
\[
c = \frac{\max\Big\{\sum_{j}\Big(a_{i_2j}-a_{i_1j}\Big)\rho_{jj}^{i_2i_2}, \sum_{j}\Big(a_{i_1j}-a_{i_2j}\Big)\rho_{jj}^{i_1i_1}\Big\}}{\Big|\sum_{j}\Big(a_{i_1j}-a_{i_2j}\Big)\rho^{i_1i_2}_{jj}\Big|}.
\]
which is just to make
\begin{equation*}
\sum_{j}\Big(a_{i_2j}-a_{i_1j}\Big)\rho_{jj}^{i_2i_2} < c \Big|\sum_{j}\Big(a_{i_1j}-a_{i_2j}\Big)\rho^{i_1i_2}_{jj}\Big|,
\end{equation*}
and
\begin{equation*}
\sum_{j}\Big(a_{i_1j}-a_{i_2j}\Big)\rho_{jj}^{i_1i_1} < c\Big|\sum_{j}\Big(a_{i_1j}-a_{i_2j}\Big)\rho^{i_2i_1}_{jj}\Big|.
\end{equation*}
Now one could choose $u_{11}=\sqrt{1-x}$, and
$u_{12}=e^{ir}\sqrt{x}$, where $x$ is a positive real number, and
$r$ is a proper real number such that
$u_{11}^*u_{12}\sum_{j}(a_{i_1j}-a_{i_2j})\rho^{i_1i_2}_{jj}$
is also a positive real number. It can be checked that if
$0<x<\frac{1}{c^2+1}$, we have
\begin{equation*}
\frac{u_{11}}{|u_{12}|}=\sqrt{\frac{1-x}{x}} > c,
\end{equation*}
which implies
\begin{eqnarray*}
Gain & = & 2|u_{11}|\cdot|u_{12}|\Big|\sum_{j}\Big(a_{i_1j}-a_{i_2j}\Big)\rho^{i_1i_2}_{jj}\Big|\\
& & - |u_{12}|^2\sum_{j}\Big(a_{i_2j}-a_{i_1j}\Big)\rho_{jj}^{i_2i_2}-|u_{21}|^2\sum_{j}\Big(a_{i_1j}-a_{i_2j}\Big)\rho_{jj}^{i_1i_1}\\
& > & 0.
\end{eqnarray*}
\end{proof}

\subsection{Eq.~\eqref{eq: sdp condition} violated: gain by an explicit POVM}\label{sec: sdp condition}
In the rest of this section, we assume that condition Eq.~\eqref{eq: eq condition} holds, because otherwise there exists explicit local unitary operation to increase the payoff.

First note that under this assumption, all matrices $B_i = \Big[\sum_j\rho^{i_1i_2}_{jj}(a_{ij}-a_{i_1j})\Big]_{i_1i_2}$ are Hermitians. Indeed, we have
\begin{align}
B_i(i_2,i_1)^* & = \sum_j (\rho_{jj}^{i_2i_1})^*(a_{ij} - a_{i_2j}) &  (\text{because } a_{ij}, a_{i_2j} \in \mbR) \\
& = \sum_j \rho_{jj}^{i_1i_2} (a_{ij} - a_{i_2j}) &  (\text{because $\rho$ is Hermitian}) \\
& = \sum_j \rho_{jj}^{i_1i_2} (a_{ij} - a_{i_1j}) &  (\text{by Eq.~\eqref{eq: eq condition}})
\end{align}
Therefore all $B_i$'s have spectral decompositions.

Now suppose that $B_i \preceq 0$ is not true for some $i\in [m]$. Without loss of generality, assume that $i=1$, namely $B_1$ has a positive eigenvalue. We denote it by $\lambda$, and the corresponding eigenvector (with unit $\ell_2$ norm) is $|\psi\rangle$.
Note that the first row contains all $0$'s, and since we assumed Eq.~\eqref{eq: eq condition}, so is the first column. This allows us to write $|\psi\rangle = (0, v_2, v_3, ..., v_m)^T$ for some $v_i$'s. For the convenience of our discussion, we suppose $|v_2|\neq 0$.
Otherwise, since $\langle\phi|B_1|\phi\rangle$ is a continuous function of $|\phi\rangle$, one can always perturb $\ket{\psi}$ a little to get another vector $|\psi'\rangle$ such that $\lambda' = \langle\psi'|B_1|\psi'\rangle > 0$ and the second entry of $\ket{\psi'}$ is not 0 (while still keeping the first entry being 0). Then we replace $|\psi\rangle$ and $\lambda$ by $|\psi'\rangle$ and $\lambda'$ in the following argument.

In the following, we will construct a local POVM $\{E_i\}$ by which Player 1 can strictly increase her payoff, which will complete the proof. Set $\{E_i\}$ to be
\begin{equation}
E_1=\varepsilon|\psi\rangle\langle\psi|+|1\rangle\langle1|=
\begin{pmatrix}
1 & 0 & 0 & \cdots & 0\\
0 & \varepsilon|v_2|^2 & \varepsilon v_2v_3^* & \cdots & \varepsilon v_2v_{m}^*\\
0 & \varepsilon v_2^*v_3 & \varepsilon|v_3|^2 & \cdots & \varepsilon v_3v_{m}^*\\
\vdots & \vdots & \vdots & \ddots & \vdots\\
0 & \varepsilon v_2^*v_{m} & \varepsilon v_3^*v_{m} & \cdots &
\varepsilon|v_m|^2
\end{pmatrix},
\end{equation}
\begin{equation}
E_2=
\begin{pmatrix}
0 & 0 & 0 & \cdots & 0\\
0 & d_{22} & -\varepsilon v_2v_3^* & \cdots & -\varepsilon v_2v_{m}^*\\
0 & -\varepsilon v_2^*v_3 & d_{23} & \cdots & 0\\
\vdots & \vdots & \vdots & \ddots & \vdots\\
0 & -\varepsilon v_2^*v_{m} & 0 & \cdots & d_{2,m}
\end{pmatrix},
\end{equation}
\begin{equation}
E_3=
\begin{pmatrix}
0 & 0 & 0 & 0 & \cdots & 0\\
0 & 0 & 0 & 0 & \cdots & 0\\
0 & 0 & d_{33} & -\varepsilon v_3v_4^* & \cdots & -\varepsilon v_3v_m^*\\
0 & 0 & -\varepsilon v_3^*v_4 & d_{34} & \cdots & 0\\
\vdots & \vdots & \vdots & \vdots & \ddots & \vdots\\
0 & 0 & -\varepsilon v_3^*v_m & 0 & \cdots & d_{3,m}
\end{pmatrix},
\end{equation}
\ \ \ \ \ \ \ \ \ \ \ \ \ \ \ \ \ \  \ \  \ \ \ \ \ \ \ \ \ \ \ \ \
\ \ \ \ \ \ \ \ $\vdots$
\begin{equation}
E_{m-1}=
\begin{pmatrix}
0 & \cdots & 0 & 0 & 0\\
\vdots & \ddots & \vdots & \vdots & \vdots\\
0 & \cdots & 0 & 0 & 0\\
0 & \cdots & 0 & d_{m-1,m-1} & -\varepsilon v_{m-1}v_m^*\\
0 & \cdots & 0 & -\varepsilon v_{m-1}^*v_m & d_{m-1,m}
\end{pmatrix},
\end{equation}
\begin{equation}
E_m=
\begin{pmatrix}
0 & \cdots & 0 & 0 & 0\\
\vdots & \ddots & \vdots & \vdots & \vdots\\
0 & \cdots & 0 & 0 & 0\\
0 & \cdots & 0 & 0 & 0\\
0 & \cdots & 0 & 0 & d_{mm}
\end{pmatrix}.
\end{equation}
Here, $\varepsilon$ is a small positive number that will be determined later. For any fixed $\varepsilon$, we will choose $d_{ij}$'s as follows. Firstly, note that we have the relationship
\begin{equation}
E_1+E_2+...+E_m=I,
\end{equation}
by which one can obtain that $d_{22}=1-\varepsilon|v_2|^2$. Let $d_{2k}=(\varepsilon|v_2v_{k}|)^2/d_{22}$, thus $d_{21}d_{2k}=(\varepsilon|v_2v_{k}|)^2$ and $E_2 \succeq 0$. After fixing $E_2$, $d_{33}$ can also be gotten by using $\sum_i E_i = I$. In general, we have
\begin{align}
& d_{ii} = 1-\epsilon|v_i|^2 - d_{2,i} - \cdots - d_{i-1,i}, \quad \forall i \geq 3, \\
& d_{ik} = \epsilon^2|v_iv_{k}|^2 / d_{ii}, \quad \forall i\geq 2, \ k \geq i+1
\end{align}
By an induction on $i$, it is not difficult to see that for any fixed $B_i$'s, for $\epsilon \rightarrow 0$, it holds that
\begin{equation}
d_{ii} = 1-\epsilon|v_i|^2 - O(\epsilon^2) = 1-O(\epsilon) > 0, \quad \forall i \geq 2
\end{equation}
and
\begin{equation}\label{eq: d_ik}
d_{ik} = \epsilon^2|v_iv_k|^2/d_{ii} = O(\epsilon^2), \quad \forall i \geq 2, \forall k\geq i+1.
\end{equation}
It can be checked that $E_1=\varepsilon|\psi\rangle\langle\psi|+|1\rangle\langle1| \succeq 0$, and every other $E_i$ is also positive because it has nonnegative diagonal entries and is actually diagonally dominant Hermitian for sufficiently small $\epsilon$. Besides, since the way we defined $\{E_i\}$ satisfies $\sum_i E_i = I$, $\{E_i\}$ is a legal POVM.

Next we calculate the gain of the Player by using $\{E_i\}$ as in Lemma \ref{lem: gain}. Since $\bra{\psi} B_1 \ket{\psi} = \lambda$, we have
\begin{equation}
Tr(E_1 B_1) = \bra{1}B_1\ket{1} + \bra{\psi}B_1\ket{\psi} = \bra{\psi}B_1\ket{\psi} = \epsilon \lambda.
\end{equation}
For $i = 2, ..., m$, note that the only nonzero off-diagonal entries of $E_i$ are on the $i$-th row and column, but those entries in $B_i$ are zero. So only the diagonal entries of $E_i$ and $B_i$ contribute to $Tr(E_i^T B_i)$, and the contribution is $\sum_{k=i}^m d_{i,k}(\sum_j \rho_{jj}^{kk}(a_{ij}-a_{kj}))$. Therefore,
\begin{align}
Gain = \sum_i Tr(E_i B_i) & = \varepsilon\lambda + \sum_{i=2}^{m} \sum_{k=i}^{m} d_{i,k} \Big(\sum_j\rho^{kk}_{jj} (a_{ij}-a_{kj})\Big) \\
& = \varepsilon\lambda + \sum_{i=2}^{m} \sum_{k=i+1}^{m} d_{i,k} \Big(\sum_j\rho^{kk}_{jj} (a_{ij}-a_{kj})\Big) & (\text{because } B_i(i,i) = 0) \\
& = \varepsilon\lambda \pm O(\epsilon^2) & (\text{because of Eq.}~\eqref{eq: d_ik})
\end{align}
Here note that $m$ and $B_i$ are all fixed and only $\epsilon$ approaches to 0. So for sufficiently small $\epsilon$, the gain is strictly positive.

\section{Upper bounds for the gain}\label{sec: ub}

In the above sections, we have shown how to determine whether a
given quantum state is a QCE or not. In this section, we consider
those quantum states that are not QCE. According to the definition
of QCE, one can find a proper POVM $\{E_i\}$ such that some player
can get a strictly positive payoff gain by this operation. A natural
question is, how much is the maximal gain? In the following theorem
we provide two simple upper bounds as first-step attempts.

\begin{Thm}
Suppose that the maximum eigenvalue of $B_i = \big[\sum_j\rho^{i_1i_2}_{jj} (a_{ij}-a_{i_1j})\big]_{i_1i_2}$ is $\lambda_i$. Let $\lambda = \max_i \lambda_i$, then we have
\begin{equation}\label{eq:Gupper}
Gain \leq m\lambda.
\end{equation}
\end{Thm}
\begin{proof}
Since $B_i\preceq \lambda_iI$ and each $E_i$ is positive, it holds that
\begin{equation*}
Tr(E_i B_i) \leq Tr(E_i (\lambda_iI)) = \lambda_i Tr(E_i).
\end{equation*}
Thus,
\begin{align}
Gain & = \sum_i Tr(E_i B_i) \\
& \leq \sum_i \lambda_i Tr(E_i) \\
& \leq \lambda \sum_i Tr(E_i) & (Tr(E_i) \geq 0) \\
& = m \lambda & (\sum_i Tr(E_i) = Tr(\sum_i E_i) = m).
\end{align}
\end{proof}

Another bound is the following.
\begin{Thm}
Suppose that the eigenvalues of $B_i = \big[\sum_j\rho^{i_1i_2}_{jj} (a_{ij}-a_{i_1j})\big]_{i_1i_2}$ is $\lambda_{i1}, ..., \lambda_{im}$. Then
\begin{equation}
Gain \leq \sum_{ij: \lambda_{ij}>0} \lambda_{ij}.
\end{equation}
\end{Thm}
\begin{proof}
Suppose the spectral decomposition of $B_i$ is $B_i=\sum_{j\in [m]}
\lambda_{ij}\ket{\psi_{ij}}\bra{\psi_{ij}}$. Then
\begin{align}
Gain & = \sum_i Tr(E_i B_i) \\
& = \sum_i Tr\big(E_i \sum_{j\in [m]} \lambda_{ij}\ket{\psi_{ij}}\bra{\psi_{ij}} \big) \\
& \leq \sum_{ij: \lambda_{ij} > 0} \lambda_{ij} \bra{\psi_{ij}} E_i \ket{\psi_{ij}}\\
& \leq \sum_{ij: \lambda_{ij} > 0} \lambda_{ij} & (\bra{\psi_{ij}} E_i \ket{\psi_{ij}} \leq \bra{\psi_{ij}} I \ket{\psi_{ij}} = 1)
\end{align}
\end{proof}

The above bounds can be pretty tight. By the dual SDP, it is not
hard to see that the gain is the following value:
\begin{align*}
    \min & \quad \emph{tr}(Y)\\
    \emph{s.t.} & \quad Y \succeq B_i,\ \forall i\in [m]
\end{align*}
Consider the following example:
\[
u_1 = \begin{bmatrix} 1 & \cdots & 1 \\ 0 & \cdots & 0 \\ \vdots & \ddots & \vdots \\ 0 & \cdots & 0 \end{bmatrix}_{m\times n}, \quad
\rho = \begin{bmatrix} \frac{1}{mn} & &  \\ & \ddots & \\  & & \frac{1}{mn} \end{bmatrix}_{mn\times mn}
\]
It is not hard to verify that
\[B_1 = diag(0,1/m, ..., 1/m), \quad B_2 = ... = B_m = diag(-1/m, 0, ..., 0).\]
Therefore, the gain is $tr(B_1) = (m-1)/m$, which matches the second bound, and is also close to the first bound $1$.

\section{Open problems}\label{sec: open problems}
Some open problems are left for future explorations.
\begin{enumerate}
\item In Section 4, we show that if the condition is not satisfied for Player $i$, the player can use a POVM to obtain a strictly positive gain. A natural questions is, can the POVM be replaced by a unitary operation? In general, can the maximum gain always be achieved by a unitary operation?

\item Can the condition be simplified if $\rho$ is a pure state?


\item How to improve the bounds in Section \ref{sec: ub}?

\item Can we have a nice characterization of $\epsilon$-approximate QCE? (Results in Section \ref{sec: ub} provide sufficient conditions for $\epsilon$-QCE, where the $\epsilon$ is the one of the given upper bounds. We hope to say more.)
\end{enumerate}

\vspace{0.1in}

\subsection*{Acknowledgment} Z.W. thanks Kewk Leong Chuan, Ji Zhengfeng and Iordanis Kerenidis for helpful comments. Z.W. was supported by the grant from the Centre for Quantum Technologies (CQT), the WBS grants under contracts no. R-710-000-008-271 and R-710-000-007-271. S.Z. was supported by China Basic Research Grant 2011CBA00300 (sub-project 2011CBA00301), and Research Grants Council of the Hong Kong S.A.R. (Project no. CUHK419309 and CUHK418710). Part of the work was done when S.Z. visited CQT and Tsinghua University, the latter under the support of China Basic Research Grant 2007CB807900 (sub-project 2007CB807901).
\bibliographystyle{alpha}
\bibliography{QCE}

\newcommand{\etalchar}[1]{$^{#1}$}
\begin{thebibliography}{DLX{\etalchar{+}}02b}

\bibitem[Aum74]{Aum74}
Robert Aumann.
\newblock Subjectivity and correlation in randomized strategies.
\newblock {\em Journal of Mathematical Economics}, 1:67--96, 1974.

\bibitem[BH01a]{BH01}
Simon Benjamin and Patrick Hayden.
\newblock Comment on ``quantum games and quantum strategies".
\newblock {\em Physical Review Letters}, 87(6):069801, 2001.

\bibitem[BH01b]{BH01b}
Simon Benjamin and Patrick Hayden.
\newblock Multiplayer quantum games.
\newblock {\em Physical Review A}, 64(3):030301, 2001.

\bibitem[CT06]{CT06}
Taksu Cheon and Izumi Tsutsui.
\newblock Classical and quantum contents of solvable game theory on hilbert
  space.
\newblock {\em Physics Letters A}, 348:147--152, 2006.

\bibitem[DLX{\etalchar{+}}02a]{DLX+02b}
Jiangfeng Du, Hui Li, Xiaodong Xu, Mingjun Shi, Jihui Wu, Xianyi Zhou, and
  Rongdian Han.
\newblock Experimental realization of quantum games on a quantum computer.
\newblock {\em Physical Review Letters}, 88(5-6):137902, 2002.

\bibitem[DLX{\etalchar{+}}02b]{DLX+02a}
Jiangfeng Du, Hui Li, Xiaodong Xu, Xianyi Zhou, and Rongdian Han.
\newblock Entanglement enhanced multiplayer quantum games.
\newblock {\em Physics Letters A}, 302(5-6):229--233, 2002.

\bibitem[EWL99]{EWL99}
Jens Eisert, Martin Wilkens, and Maciej Lewenstein.
\newblock Quantum games and quantum strategies.
\newblock {\em Physical Review Letters}, 83(15):3077--3080, 1999.

\bibitem[FA03]{FA03}
Adrian Flitney and Derek Abbott.
\newblock Advantage of a quantum player over a classical one in 2 x 2 quantum
  games.
\newblock {\em Proceedings of The Royal Society A: Mathematical, Physical and
  Engineering Sciences}, 459(2038):2463--2474, 2003.

\bibitem[FA05]{FA05}
Adrian Flitney and Derek Abbott.
\newblock Quantum games with decoherence.
\newblock {\em Journal of Physics A: Mathematical and General}, 38(2):449--459,
  2005.

\bibitem[FT91]{FT91}
Drew Fudenberg and Jean Tirole.
\newblock {\em Game theory}.
\newblock MIT Press, 1991.

\bibitem[LJ03]{LJ03}
Chiu~Fan Lee and Neil Johnson.
\newblock Efficiency and formalism of quantum games.
\newblock {\em Physical Review A}, 67:022311, 2003.

\bibitem[Nas51]{Nas51}
John Nash.
\newblock Non-cooperative games.
\newblock {\em The Annals of Mathematics}, 54(2):286--295, 1951.

\bibitem[NC00]{NC00}
Michael Nielsen and Isaac Chuang.
\newblock {\em Quantum Computation and Quantum Information}.
\newblock Cambridge University Press, Cambridge, UK, 2000.

\bibitem[OR94]{OR94}
Martin Osborne and Ariel Rubinstein.
\newblock {\em A course in game theory}.
\newblock MIT Press, 1994.

\bibitem[PSWZ07]{PSWZ07}
Robert Prevedel, Andr{\' e} Stefanov, Philip Walther, and Anton Zeilinger.
\newblock Experimental realization of a quantum game on a one-way quantum
  computer.
\newblock {\em New Journal of Physics}, 9:205, 2007.

\bibitem[VNRT07]{VNRT07}
Vijay Vazirani, Noam Nisan, Tim Roughgarden, and {\' E}va Tardos.
\newblock {\em Algorithmic Game Theory}.
\newblock Cambridge University Press, 2007.

\bibitem[Wat08]{Wat08}
John Watrous.
\newblock {\em Theory of Quantum Information}.
\newblock Lecture notes, University of Waterloo, 2008.

\bibitem[Zha10]{Zha10}
Shengyu Zhang.
\newblock Quantum strategic game theory.
\newblock {\em arXiv:1012.5141}, 2010.
\newblock Presentation in The 14th International Workshop on Quantum
  Information Processing (QIP'11).

\end{thebibliography}
\end{document}